\documentclass{llncs} 
\usepackage{latexsym,amssymb,amsmath,graphicx,bbm}
\usepackage{epic}
\usepackage{psfrag}
\usepackage{amsbsy}
\usepackage{mathptm} 
\usepackage[bbgreekl]{mathbbol}
\usepackage{cmgreek} 
\usepackage{epsfig}
\usepackage{cite}
\usepackage[small,normal,bf,up]{caption2}

\DeclareMathOperator{\prob}{{\text{\rm P}}}
\newcommand{\dr}{{\tt r}}
\newcommand{\D}[2]{\frac{\partial #1}{\partial #2}}
\newcommand{\un}[1]{\underline{#1}}
\newcommand{\brc}[1]{\left({#1}\right)}

\newcommand{\abs}[1]{\left\lvert#1\right\rvert}
\newcommand{\norm}[1]{\left\lVert#1\right\rVert}

\newcommand{\nae}{\text{\tiny NAE}}

\newcommand{\openone}{\leavevmode\hbox{\small1\normalsize\kern-.33em1}}

\begin{document}
\renewcommand{\textfraction}{0}

\title{\Large{Bounds on Threshold of Regular Random $k$-SAT}}
\author{Vishwambhar Rathi\inst{1,}\inst{2} \and Erik Aurell\inst{1,}\inst{3} \and Lars Rasmussen\inst{1,}\inst{2} \and Mikael Skoglund\inst{1,}\inst{2} 
\thanks{Email: vish@kth.se, eaurell@kth.se, lars.rasmussen@ee.kth.se, skoglund@ee.kth.se.}} 
\institute{KTH Linnaeus Centre ACCESS, KTH-Royal Institute of Technology, Stockholm, Sweden \and School of Electrical Engineering, KTH-Royal Institute of Technology, Stockholm, Sweden \and Dept. Information and Computer Science, TKK-Helsinki University of Technology, Espoo, Finland}
\institute{KTH Linnaeus Centre ACCESS, KTH-Royal Institute of Technology, Stockholm, Sweden \and School of Electrical Engineering, KTH-Royal Institute of Technology, Stockholm, Sweden \and Dept. Information and Computer Science, TKK-Helsinki University of Technology, Espoo, Finland}

\maketitle

\begin{abstract}
We consider the regular model of formula generation in conjunctive normal form
(CNF) introduced by Boufkhad et. al. in \cite{BDIS05}. In \cite{BDIS05}, it was
shown that the threshold for regular random $2$-SAT  is equal to unity. Also,
upper and lower bound on the threshold for regular random $3$-SAT were derived.
Using the first moment method, we derive an upper bound on the threshold for
regular random $k$-SAT for any $k \geq 3$ and show that for large $k$ the
threshold is upper bounded by $2^k \ln (2)$.  We also derive upper bounds on
the threshold for Not-All-Equal (NAE) satisfiability for $k \geq 3$ and show that for large $k$, the
NAE-satisfiability threshold is upper bounded by $2^{k-1} \ln(2)$. For both 
satisfiability and NAE-satisfiability, the obtained upper bound 
matches with the corresponding bound for the uniform model of formula
generation \cite{FrP83, AcM06}. 


For the uniform model, in a series of break through papers Achlioptas, Moore,
and Peres showed that a careful application of the second moment method yields
a significantly better lower bound on threshold as compared to any rigorously
proven algorithmic bound \cite{AcP04, AcM06}. The second moment method shows
the existence of a satisfying assignment with uniform positive probability
(w.u.p.p.).  Thanks to the result of Friedgut for uniform model \cite{Fri99},
existence of a satisfying assignment w.u.p.p.~translates to existence of a
satisfying assignment with high probability (w.h.p.). Thus, the second moment
method gives a lower bound on the threshold. As there is no known Friedgut type
result for regular random model, we assume that for regular random model
existence of a satisfying assignments w.u.p.p.~translates to existence of a
satisfying assignments w.h.p. We derive the second moment of
the number of satisfying assignments for regular random $k$-SAT for $k \geq 3$.
There are two aspects in deriving the lower bound using the second moment
method. The first aspect is given any $k$, numerically evaluate the lower
bound on the threshold. The second aspect is to derive the lower bound as a
function of $k$ for large enough $k$.  We address the first aspect and evaluate
the lower bound on threshold.  The numerical evaluation suggests that as $k$
increases the obtained lower bound on the satisfiability threshold of a regular
random formula converges to the lower bound obtained for the uniform model.
Similarly, we obtain lower bounds on the NAE-satisfiability threshold of the
regular random formulas and observe that the obtained lower bound seems to
converge to the corresponding lower bound  for the uniform model as $k$
increases. 
\end{abstract}

\section{Regular Formulas and Motivation}

A clause is a disjunction (OR) of $k$ variables. A formula is a conjunction
(AND) of a finite set of clauses. A $k$-SAT formula is a formula where each
clause is a disjunction of $k$ literals.  A {\it legal} clause is one in which
there are no repeated or complementary literals. Using the terminology of
\cite{BDIS05}, we say that a formula is {\it simple} if it consists of only
legal clauses. A {\it configuration} formula is not necessarily legal. A
satisfying (SAT) assignment of a formula is a truth assignment of variables for
which the formula evaluates to true. A Not-All-Equal (NAE) satisfying
assignment is a truth assignment such that every clause is connected to at
least one true literal and at least one false literal. We denote the number of
variables by $n$, the number of clauses by $m$, and the clause density, i.e.
the ratio of clauses to variables, by $\alpha = \frac{m}{n}$.  We denote the
binary entropy function by $h(\cdot)$, $h(x) \triangleq -x \ln(x) - (1-x)
\ln(1-x)$, where the logarithm is the natural logarithm.

The popular, uniform k-SAT model generates a formula by selecting uniformly and
independently $m$-clauses from the set of all $2^k \binom{n}{k}$ $k$-clauses. In
this model, the literal degree can vary. We are interested in the model where
the literal degree is almost constant, which was introduced in
\cite{BDIS05}. Suppose each literal has degree $r$.  Then $2 n r = k m$, 
which gives $\alpha = 2 r / k$. Hence $\alpha$ can only take values from a 
discrete set of possible values.  To circumvent this, we allow each literal to 
take two possible values for a degree.  For a given $\alpha$, let $r =  \frac{k
\alpha}{2}$ and $\dr = \lfloor r \rfloor$.  Each literal has degree either
$\dr$ or $\dr + 1$. Also a literal and its negation have the same degree. Thus, we
can speak of the degree of a variable which is the same as the degree of its
literals. Let the number of variables with degree $d$ be $n_d$, $d \in \{\dr,
\dr+1\}$. Let $X_1,\dots,X_{n_{\dr}}$ be the variables which have degree $\dr$
and $X_{n_{\dr} + 1}, \dots, X_n$ be the variables with degree $\dr+1$. Then, 
\[
n_{\dr} = n + \dr n - \left\lfloor\frac{k \alpha n}{2}\right \rfloor, \quad n_{\dr + 1} = \left\lfloor \frac{k \alpha n}{2} \right\rfloor - \dr n
\]
and $n_\dr + n_{\dr+1} = n$. 
As we are interested in the asymptotic setting, we will ignore the floor in the sequel. 
We denote the fraction of variables with degree $\dr$ (resp. $\dr+1$) by 
$\Lambda_\dr$ (resp. $\Lambda_{\dr+1}$) which is given by
\begin{equation}\label{eq:defLambda}
\Lambda_\dr = 1+\dr-\frac{k \alpha}{2}, \quad \Lambda_{\dr+1} = \frac{k \alpha}{2} - \dr.
\end{equation}
When $\Lambda_\dr$ or $\Lambda_{\dr+1}$ is zero, we refer to such formulas as
{\it strictly} regular random formulas. This implies that there is no variation
in literal degree. If $\Lambda_{\dr}, \Lambda_{\dr+1} > 0$, then we say that
the formulas are $2$-regular random formulas. 

A formula is represented by a
bipartite graph. The left vertices represent the literals and right vertices
represent the clauses. A literal is connected to a clause if it appears in the
clause.  There are $k \alpha n$ edges coming out from all the literals and $k
\alpha n$ edges coming out from the clauses. We assign the labels from the set
$\mathcal{E}= \{1, \dots, k \alpha n\}$ to edges on both sides of the 
bipartite graph. In order to generate a formula, we generate a random
permutation $\Pi$ on $\mathcal{E}$. Now we connect an edge $i$ on the literal 
node side to an edge $\Pi(i)$ on the clause node side. This gives rise to a regular
random k-SAT formula. Note that not all the formulas generated by this
procedure are simple. However, it was shown in \cite{BDIS05} that the
threshold is the same for this collection of formulas and the collection of simple
formulas. Thus, we can work with the collection of configuration formulas
generated by this procedure. 

The regular random $k$-SAT formulas are of interest because such instances are
computationally harder than the uniform $k$-SAT instances. This was
experimentally observed in \cite{BDIS05}, where the authors also derived upper
and lower bounds for regular random $3$-SAT. The upper bound was derived using
the first moment method. The lower bound was derived by analyzing a greedy algorithm 
proposed in \cite{KKL02}.  To the best of our knowledge, there are no known
upper and lower bounds on the thresholds for regular random formulas for $k
 > 3$. 

Using the first moment method, we compute an upper bound $\alpha_u^*$ on the
satisfiability threshold $\alpha^*$ for regular random formulas for $k \geq 3$.
We show that $\alpha^* \leq 2^k \ln(2)$, which coincides with the upper bound
for the uniform model. We also apply the first moment method to obtain an upper bound 
$\alpha_{u, \text{\tiny NAE}}^*$ on the NAE-satisfiability threshold
$\alpha_{\text{\tiny NAE}}^*$ of regular random formulas. We show that
$\alpha_{\text{\tiny NAE}}^* \leq 2^{k-1} \ln(2)$ which coincides with the
corresponding bound for the uniform model.  

In order to derive a lower bound $\alpha_l^*$ on the threshold, we apply the
second moment method to the number of satisfying assignments. The second moment
method shows the existence of a satisfying assignment with uniform positive 
probability (w.u.p.p.).  Due to the result of Friedgut for uniform model
\cite{Fri99}, existence of a satisfying assignment w.u.p.p.~translates to
existence of a satisfying assignment with high probability (w.h.p.). Thus, the
second moment method gives lower bound on the threshold for uniform model. As
there is no known Friedgut type result for regular random model, we assume that
for regular random model existence of a satisfying assignments
w.u.p.p.~translates to existence of a satisfying assignments w.h.p. This
permits us to say that second moment method gives valid lower bound on the
threshold. We compute the second moment of the number of satisfying assignments
for regular random model. Similar to the case of the uniform model, we show
that for the second moment method to succeed the term corresponding to overlap
$n/2$  should dominate other overlap terms. We observe that the obtained lower
bound $\alpha_l^*$ converges to the corresponding lower bound of the uniform
model, which is $2^k \ln (2) - (k+1)\frac{\ln(2)}{2} - 1$ as $k$  increases.
Similarly, by computing the second moment of the number of NAE-satisfying
assignments we obtain that $\alpha_{l,\text{\tiny NAE}}$ converges to the
corresponding bound $2^{k-1} \ln(2) - O(1)$  for the uniform model. The lower
bounds are not obtained explicitly as computing the second moment requires
finding all the positive solutions of a system of polynomial equations. For
small values of $k$, this can be done exactly. However, for large values of $k$
we resort to a numerical approach. Our main contribution is that we obtain
almost matching lower and upper bounds on the satisfiability (resp.
NAE-satisfiability) threshold for the regular random formulas.  Thus, we answer
in affirmative the following question posed in \cite{AcM06}: {\it Does the
second moment method perform well for problems that are symmetric ``on
average''? For example, does it perform well for regular random $k$-SAT where
every literal appears an equal number of times?}.

In the next section, we obtain an upper bound on the satisfiability threshold and
NAE satisfiability threshold. 

\section{Upper Bound on Threshold via First Moment}
Let $X$ be a non-negative integer-valued random variable and $E(X)$ be its expectation. 
Then the first moment method gives: $\prob\brc{X > 0} \leq E(X)$. 
Note that by choosing $X$ to be the number of solutions of a random formula, we can 
obtain an upper bound on the threshold $\alpha^*$ beyond which no solution exists 
with probability one. This upper bound corresponds to the largest value of $\alpha$ at which 
 the average number of solutions goes to zero as $n$ tends to infinity. 
In the following lemma, we derive the first moment of the number of SAT solutions  
of the regular random $k$-SAT for $k \geq 3$. 
\begin{lemma}
Let $N(n, \alpha)$ (resp. $N_{\text{\tiny NAE}}(n, \alpha)$) be the number of satisfying (resp. NAE satisfying) 
assignments for a randomly generated regular k-SAT 
formula. Then\footnote{We assume that $k \alpha n$ is an even integer.}, 
\begin{equation}\label{eq:moment1}
E(N(n, \alpha)) =  2^n \frac{\brc{\brc{\frac{k \alpha n}{2}}!}^2}{(k \alpha n)!} 
\mathrm{coef}\brc{\brc{\frac{p(x)}{x}}^{\alpha n}, x^{\frac{k \alpha n}{2}-\alpha n}}, 
\end{equation}
\begin{equation}\label{eq:moment1nae}
E(N_{\text{\tiny NAE}}(n, \alpha)) = 2^n \frac{\brc{\brc{\frac{k \alpha n}{2}}!}^2}{(k \alpha n)!} 
\mathrm{coef}\brc{\brc{\frac{p_{\text{\tiny NAE}}(x)}{x}}^{\alpha n}, x^{\frac{k \alpha n}{2}-\alpha n}}, 
\end{equation}
where 
\begin{equation}\label{eq:defpx}
	p(x) = (1+x)^k - 1, \quad p_{\text{\tiny NAE}}(x) = (1+x)^k-1-x^k,  
\end{equation}
and $\mathrm{coef}\brc{p(x)^{\alpha n}, x^{\frac{k \alpha n}{2}}}$ denotes the coefficient of $x^{\frac{k \alpha n}{2}}$ in the 
expansion of $p(x)^{\alpha n}$. 
\end{lemma}
\begin{proof} 
Due to symmetry of the formula generation, any assignment of variables has the same 
probability of being a solution. This implies
\[
E(N(n, \alpha)) = 2^n \prob\brc{X=\{0,\dots,0\}\text{ is a solution}}. 
\]
The probability of the all-zero vector being a solution is given by 
\begin{multline*}
\prob\brc{X=\{0,\dots,0\}\text{ is a solution}} = \\ \frac{\text{Number of formulas for which } X=\{0,\dots,0\} 
\text{ is a solution}}{\text{Total number of formulas}}.
\end{multline*}
The total number of formulas is given by $(k \alpha n)!$. The total number of formulas for which the all-zero  
 assignment is a solution is given by
\[
  \brc{\brc{\frac{k \alpha n}{2}}!}^2 \mathrm{coef}\brc{p(x)^m, x^{\frac{k \alpha n}{2}}}. 
\] 
The factorial terms correspond to permuting the edges among true and false literals. Note that there are equal 
numbers of true and false literals. The generating function $p(x)$ corresponds to placing at least one positive literal 
in a clause. With these results and observing that 
\[
\mathrm{coef}\brc{p(x)^{\alpha n}, x^{\frac{k \alpha n}{2}}} = 
\mathrm{coef}\brc{\brc{\frac{p(x)}{x}}^{\alpha n}, x^{\frac{k \alpha n}{2}-\alpha n}}, 
\]
we obtain (\ref{eq:moment1}). The derivation for $E(N_{\text{\tiny NAE}}(n, \alpha))$ is identical except 
that the generating function for clauses is given by $p_{\text{\tiny NAE}}(x)$.
\end{proof}

We now state the Hayman method to approximate the coef-term which is
asymptotically correct \cite{Gar95}. 
\begin{lemma}[Hayman Method]\label{lem:hayman} 
Let $q(y) = \sum_i q_i y^i$ be a polynomial with non-negative coefficients such that 
$q_0 \neq 0$ and $q_1 \neq 0$. Define 
\begin{equation}\label{eq:defaqbq}
   a_q(y) = y \frac{d q(y)}{d y} \frac{1}{q(y)}, \quad b_q(y) = y \frac{d a_q(y)}{d y}. 
\end{equation}
Then, 
\begin{equation}\label{eq:hayman}
	\mathrm{coef}\brc{q(y)^n, y^{\omega n}} = \frac{q(y_\omega)^n}{(y_\omega)^{\omega n} 
\sqrt{2 \pi n b_q(y_\omega)}} (1+o(1)), 
\end{equation}
where $y_\omega$ is the unique positive solution of the saddle point equation $a_q(y) = \omega$. 
\end{lemma}

We now use Lemma \ref{lem:hayman} to compute the expectation of the total number of solutions. 
\begin{lemma}\label{lem:moment1hayman}
Let $N(n, \alpha)$ (resp. $N_{\text{\tiny NAE}}(n, \alpha)$) denote the total number of satisfying 
(resp. NAE satisfying) assignments of a regular random k-SAT 
formula. Let $q(x)=\frac{p(x)}{x}$, \\ $q_{\text{\tiny NAE}}(x) = \frac{p_{\text{\tiny NAE}}}{x}$, where $p(x)$ and 
$p_{\text{\tiny NAE}}(x)$ is defined in (\ref{eq:defpx}). Then, 
\begin{equation}\label{eq:moment1hayman}
E(N(n, \alpha)) = \sqrt{\frac{k}{4 b_q(x_k)}} 
e^{n \brc{\ln(2) - k \alpha \ln(2) + \alpha \ln\brc{q(x_k)}-\brc{\frac{k \alpha}{2}-\alpha} \ln\brc{x_k}}} (1+o(1)), 
\end{equation}
\begin{equation}\label{eq:moment1haymannae}
E(N_{\text{\tiny NAE}}(n, \alpha)) = \frac{\sqrt{k}
e^{n \brc{\ln(2) (1 - k \alpha) + \alpha \ln\brc{q_{\text{\tiny NAE}}(x_{k, \text{\tiny NAE}})}-\brc{\frac{k \alpha}{2}-\alpha} \ln\brc{x_{k, \text{\tiny NAE}}}}}}{\sqrt{4 b_{q_{\text{\tiny NAE}}}(x_{k, \text{\tiny NAE}})}} (1+o(1)), 
\end{equation}
where $x_k$ (resp. $x_{k, \text{\tiny NAE}}$) is the positive solution of $a_q(x) = \frac{k}{2}-1$ 
(resp. $a_{q_{\text{\tiny NAE}}}(x) = \frac{k}{2}-1$). The quantity $a_q(x)$, $a_{q_{\text{\tiny NAE}}}(x)$, 
 $b_q(x)$, and $b_{q_{\text{\tiny NAE}}}(x)$ are defined according to (\ref{eq:defaqbq}).
\end{lemma}
In the following lemma we derive explicit upper bounds on the satisfiability and NAE satisfiability thresholds for 
$k \geq 3$.
\begin{lemma}[Upper bound]
Let $\alpha^*$ (resp. $\alpha_{\text{\tiny NAE}}^*$) be the satisfiability
(resp. NAE satisfiability) threshold for the regular random $k$-SAT formulas. Define
$\alpha_u^*$ (resp. $\alpha_{u, \text{\tiny NAE}}^*$) to be the upper bound on
$\alpha^*$ (resp. $\alpha_{\text{\tiny NAE}}^*$) obtained by the first moment
method. Then,
\begin{equation}\label{eq:ubalpha}
\alpha^* \leq  \alpha_u^* \leq 2^k \ln(2) (1+o_k(1)),  \quad 
\alpha_{\text{\tiny NAE}}^* \leq \alpha_{u, \text{\tiny NAE}}^* = 
2^{k-1} \ln(2)-\frac{\ln(2)}{2}-o_k(1).
\end{equation}
\end{lemma}
\begin{proof}
 We observe that the solution $x_k$ of the saddle point equation $a_q(x) = \frac{k}{2}-1$  
satisfies: $x_k = {\text{argmin}}_{x > 0} \frac{q(x)}{x^{\frac{k}{2}-1}}$, where $a_q(x)$ is defined 
according to (\ref{eq:defaqbq}).  
This implies that we obtain the following upper bound on the growth rate of $E(N(n, \alpha)))$ for any 
$x > 0$, 
\begin{equation}\label{eq:ubENx} 
\lim_{n \to \infty} \frac{\ln\brc{E(N(n, \alpha))}}{n} \leq \ln(2) - k \alpha \ln(2) + \alpha \ln(q(x)) - 
\brc{\frac{k \alpha}{2}-\alpha} \ln(x).
\end{equation}
We substitute $x=1-\frac{1}{2^k}$ in (\ref{eq:ubENx}). Then we use the series 
expansion of $\ln(1-x)$, $1/i \geq 1/2^i$, and $-1/i \geq -1$
 to obtain the following upper bound on the threshold, 
\begin{equation}\label{eq:temp4}
\alpha^* \leq \frac{2^k \ln(2)}{\frac{1}{\brc{1-\frac{1}{2^{k+1}}}^k} + \frac{k}{2^{k+4}} + 
\frac{1}{2^{k+2} \brc{1-\frac{1}{2^{k+1}}}^{2 k}} - \frac{1}{2^{k+1}}. 
}
\end{equation}
The summation of the last three terms in the denominator of (\ref{eq:temp4}) is positive. This 
can be easily seen for $k \geq 8$. For $3 \leq k < 8$, it can be verified by explicit calculation. 
Dropping this summation in (\ref{eq:temp4}), we obtain the desired upper bound on the threshold. 
To derive the bound for NAE satisfiability, we note that $x_{k, \text{\tiny NAE}}=1$ for $k \geq 3$. 
By substituting this in the exponent of $E(N_{\text{\tiny NAE}})$ and equating it to zero, we obtain the 
desired expression for $\alpha_{u, \text{\tiny NAE}}^*$. 
\end{proof}

In the next section we use the second moment method to obtain lower bounds on the satisfiability and 
NAE satisfiability thresholds of regular random $k$-SAT. 
\section{Second Moment}
A lower bound on the threshold can be obtained by the second moment method.  
 The second moment method is governed by the following equation
\begin{equation}\label{eq:smmethod}
  \prob\brc{X > 0} \geq \frac{E(X)^2}{E(X^2)}. 
\end{equation}
In this section we compute the second moment of $N(n, \alpha)$ and
$N_{\text{\tiny NAE}}(n, \alpha)$.  Our computation of the second moment is
inspired by the computation of the second moment for the weight and stopping set
distributions of regular LDPC codes in \cite{Vra06, Vra08} (see also \cite{BaB05}). 
We compute the second moment in the next lemma. 
\begin{lemma} 
Let $N(n, \alpha)$ be the number of satisfying solutions to a regular random  
k-SAT formula. Define the function $f(x_1, x_2, x_3)$ by 
\begin{equation}\label{eq:deff}
f(x_1, x_2, x_3) = (1+x_1+x_2+x_3)^k - (1+x_1)^k - (1+x_3)^k + 1.
\end{equation}
If the regular random formulas are strictly regular, then
\begin{multline}\label{eq:moment2deg1}
E\brc{N(n, \alpha)^2} = \\ \sum_{i=0}^n 2^n \binom{n}{i} \frac{\brc{\brc{\dr(n -i )}!}^2 \brc{(\dr i)!}^2 
\mathrm{coef}\brc{f(x_1, x_2, x_3)^{\alpha n}, x_1^{\dr (n-i)} x_2^{\dr i} x_3^{\dr (n-i)}}}{(k \alpha n)!}.
\end{multline}
If the regular random formulas are $2$-regular, then  
\begin{multline}\label{eq:moment2}
E\brc{N(n, \alpha)^2} = \sum_{i_\dr=0}^{n_{\dr}} 
\sum_{i_{\dr+1} = 0}^{n_{\dr+1}} 2^n \binom{n_\dr}{i_\dr} 
\binom{n_{\dr+1}}{i_{\dr+1}} \brc{\brc{\frac{k \alpha n}{2} - \dr i_\dr - (\dr+1) i_{\dr+1}}!}^2 \\
\frac{\brc{\brc{\dr i_\dr + (\dr+1) i_{\dr+1}}!}^2}{(k \alpha n)!} 
\mathrm{coef}\brc{f(x_1, x_2, x_3)^{\alpha n}, 
(x_1 x_3)^{\frac{k \alpha n}{2}-\dr i_\dr - (\dr+1) i_{\dr+1}} x_2^{\dr i_\dr + (\dr+1) i_{\dr+1}}}.
\end{multline}
For both the strictly regular and the 2-regular case, the expression for
$E\brc{N_{\text{\tiny NAE}}(n, \alpha)^2}$ is the same as that for $E\brc{N(n,
\alpha)^2}$ except replacing the generating function $f(x_1, x_2, x_3)$  by
$f_{\text{\tiny NAE}}(x_1, x_2, x_3)$, which is given by
\begin{multline}\label{eq:fnae}
f_{\nae}(x_1, x_2, x_3) = (1+x_1+x_2+x_3)^k - \\ \brc{(1+x_1)^k + (1+x_3)^k - 1 + (x_1 + x_2)^k - x_1^k + (x_2 + x_3)^k - x_2^k - x_3^k}.
\end{multline}
\end{lemma}
\begin{proof}
Let $\openone_{X Y}$ be the indicator variable which evaluates to $1$ if 
the truth assignments $X$ and $Y$ satisfy a randomly regular k-SAT formula. Then,  
\[
E(N(n, \alpha)^2) = \sum_{X, Y \in \{0, 1\}^n} E\brc{\openone_{\bf{X} \bf {Y}}}
= 2^n \sum_{Y \in \{0, 1\}^n} \prob\brc{{\bf 0} \text{ and } Y \text{ are
solutions}}. 
\]
The last simplification uses the fact that the number of formulas which are satisfied by both $X$ and $Y$ 
depends only on the number of variables on which $X$ and $Y$ agree. Thus, we 
fix $X$ to be the all-zero vector.  

We now consider the strictly regular case. The probability that the all-zero truth assignment
and the truth assignment $Y$ both are solutions of a randomly chosen regular
formula depends only on the {\it overlap}, i.e., the number of variables where
 the two truth assignments agree. Thus for a given overlap $i$, we can fix $Y$
to be equal to zero in the first $i$ variables and equal to $1$ in the 
 remaining variables.  This gives,
\begin{equation}\label{eq:deg1m2temp1}
E(N(n, \alpha)^2) = \sum_{i=0}^{n} 2^n \binom{n}{i} \prob\brc{{\bf 0} \text{ and } Y \text{ are solutions}}.
\end{equation}
In order to evaluate the probability that both ${\bf 0}$ and $Y$ are solutions 
for a given overlap $i$, we observe that there are four different types of edges 
connecting the literals and the clauses.  There are $\dr (n-i)$ {\bf type 1}
edges which are connected to true literals w.r.t. the {\bf 0} truth assignment and
false w.r.t. to the $Y$ truth assignment. The $\dr i$ {\bf type 2} edges are connected to
true literals w.r.t. both the  truth assignments.  There are $\dr (n-i)$ {\bf
type 3} edges which are connected to false literals w.r.t. the {\bf 0} truth assignment and true
literals w.r.t. to the $Y$ truth assignment.  The $\dr i$ {\bf type 4} edges are connected to
false literals w.r.t. both the  truth assignments. Let $f(x_1, x_2, x_3)$ be
the generating function counting the number of possible edge connections to a
clause, where the power of $x_i$ gives the number of edges of type $i$, $i \in \{1,
2, 3\}$. A clause is satisfied if it is connected to at least one type $2$
edge. Otherwise, it is satisfied if it is connected to at least one type $1$
and at least one type $3$ edge. Then the generating function $f(x_1, x_2, x_3)$
is given as in (\ref{eq:deff}).  Using this, we obtain
\begin{multline}\label{eq:deg1probtemp}
\prob\brc{{\bf 0} \text{ and } Y \text{ are solutions}} = \\ \frac{\brc{(\dr (n-i))!}^2 \brc{(\dr i)!}^2 
\mathrm{coef}\brc{f(x_1, x_2, x_3)^{\alpha n}, x_1^{\dr (n-i)} x_2^{\dr i} x_3^{\dr (n-i)}}}{(k \alpha n)!},
\end{multline}
where $(k \alpha n)!$ is the total number of formulas. Consider a given formula which is satisfied by both  truth 
assignments ${\bf 0}$ and $Y$. If we permute the positions of type $1$ edges on the clause side, we obtain another 
formula having ${\bf 0}$ and $Y$ as solutions. The argument holds true for the type $i$ edges, $i \in \{2, 3, 4\}$.
This explains the term $(\dr (n-i))!$ in (\ref{eq:deg1probtemp}) which corresponds to permuting the type $1$ edges (it is squared because of the 
same contribution from type $3$ edges). Similarly, $(\dr i)!^2$ corresponds to permuting type $2$ and type $4$ edges. 
Combining (\ref{eq:deg1m2temp1}) and (\ref{eq:deg1probtemp}), we obtain the desired expression for the second moment of the 
number of solutions as given in (\ref{eq:moment2deg1}).

We now consider the two regular case. Note that in this case the equivalent equation 
corresponding to (\ref{eq:deg1m2temp1}) is  
\begin{equation}\label{eq:m2temp1}
E(N(n, \alpha)^2) = \sum_{i_\dr=0}^{n_\dr} \sum_{i_{\dr+1}=0}^{n_{\dr+1}} 2^n \binom{n_\dr}{i_\dr} 
\binom{n_{\dr+1}}{i_{\dr+1}} \prob\brc{({\bf 0}, Y) \text{ is a solution}},
\end{equation}
where $i_\dr$ (resp. $i_{\dr+1}$) is the variable corresponding to the overlap between truth assignments  
${\bf 0}$ and $Y$ among variables with degree $\dr$ (resp. $\dr+1$). Similarly, the equivalent of 
(\ref{eq:deg1probtemp}) is given by 
\begin{multline}\label{eq:probtemp}
\prob\brc{{\bf 0} \text{ and } Y \text{ are solutions}} = 
\brc{\brc{\dr (n_\dr-i_\dr) + (\dr+1) (n_{\dr+1}-i_{\dr+1})}!}^2 \\ 
\times \frac{\brc{\brc{\dr i_\dr + (\dr+1) i_{\dr+1}}!}^2}{(k \alpha n)!} \\ 
\times \mathrm{coef}\brc{f(x_1, x_2, x_3)^{\alpha n}, 
(x_1 x_3)^{\dr (n_\dr-i_\dr) + (\dr+1) (n_{\dr+1}-i_{\dr+1})} x_2^{\dr i_\dr + (\dr+1) i_{\dr+1}}}.
\end{multline}
Combining (\ref{eq:m2temp1}) and (\ref{eq:probtemp}),  and observing that 
$
	\dr n_{\dr} + (\dr+1) n_{\dr+1} = \frac{k \alpha n}{2}, 
$
we obtain (\ref{eq:moment2}). The derivation of $E\brc{N_\nae(n, \alpha)^2}$ is identical except  
the generating function for NAE-satisfiability of a clause is different. This can be easily derived by 
observing that a clause is not NAE-satisfied for the following edge connections. Consider the case 
when a clause is connected to only one type of edge, then it is not NAE-satisfied. Next consider the case 
when a clause is connected to two types of edges. Then the combinations of  
type 1 and type 4, type 3 and type 4, type 1 and type 2, or type 2 and type 3 do not NAE-satisfies a clause. This 
gives the generating function $f_\nae(x_1, x_2, x_3)$ defined in (\ref{eq:fnae}). 
\end{proof}

In order to evaluate the second moment, we now present the multidimensional saddle point method 
in the next lemma \cite{BeR83}. A detailed technical exposition of the multidimensional saddle point method 
can be found in Appendix D of \cite{RiU08}.  
\begin{theorem} \label{thm:multisaddle}
Let $\un{i}:=(i_1, i_2 , i_3)$, $\un{j}:=(j_1, j_2, j_3)$,  and $\un{x} = (x_1, x_2, x_3)$
\[
0 < \lim_{n \to \infty} i_1/n, \quad 0 < \lim_{n \to \infty} i_{2}/n, \quad 0 < \lim_{n \to \infty} i_3/n.
\]
Let further $f(\un{x})$ be as defined in (\ref{eq:deff}) and $\un{t}=(t_1,t_2,t_3)$ be a positive solution of the saddle point equations  
$
a_f(\un{x}) \triangleq \left\{ x_i \frac{\partial \ln(f(x_1, x_2, x_3))}{\partial x_i} \right\}_{i=1}^3 = \frac{\un{i}}{\alpha n}. 
$ 
Then $\mathrm{coef}\brc{f(\un{x})^{\alpha n},\un{x}^{\un{i}}}$ can be approximated as ,
\[
\mathrm{coef}\brc{f(\un{x})^{\alpha n},\un{x}^{\un{i}}} = \frac{f(\un{t})^{\alpha n}}{({\un{t}})^{\un{i}}\sqrt{(2 \pi \alpha n)^3 \abs{B(\un{t})}}}(1+o(1)),
\]
using the saddle point method for multivariate polynomials, 
where $B(\un{x})$ is a $3 \times 3$ matrix whose elements are given by 
$B_{i,j}=x_j \D{a_{fi}(x_1, x_2, x_3)}{x_j}=B_{j,i}$ and $a_{f i}(\un{x})$ is the $i^\text{\tiny th}$ coordinate of $a_f(\un{x})$. Also, $\mathrm{coef}\brc{f(\un{x})^{\alpha n}, \un{x}^{\un{j}}}$ can be approximated in terms of $\mathrm{coef}\brc{f(\un{x})^{\alpha n},\un{x}^{\un{i}}}$. This approximation is called the   
{\em local limit theorem} of $\un{j}$ around $\un{i}$. Explicitly, if $\un{u} := \frac{1}{\sqrt{\alpha n}} (\un{j}-\un{i})$ and $\norm{\un{u}}=O((\ln{n})^{\frac{1}{3}})$, then  
\begin{eqnarray*}
\mathrm{coef}\brc{f(\un{x})^{\alpha n},\un{x}^{\un{j}}} & = & \un{t}^{\un{i}-\un{j}} \exp\brc{-\frac{1}{2} \un{u}\cdot B(\un{t})^{-1}\cdot\un{u}^T} 
 \mathrm{coef}\brc{f(\un{x})^{\alpha n},\un{x}^{\un{i}}} (1+o(1)). 
\end{eqnarray*}
\end{theorem} 

Because of the relative simplicity of the expression for the second moment, we
explain its computation in detail for the strictly regular case. Then we will
show how the arguments can be easily extended to the 2-regular case. The derivation for 
the NAE-satisfiability is identical for both cases.  
\begin{theorem}\label{thm:reglowerbound}
Consider the strictly regular random $k$-SAT model with literal degree $\dr$. Let $S(i)$ denote the 
$i^{\text{th}}$ summation term in (\ref{eq:moment2deg1}), and $\gamma=i/n$. If $S(n/2)$ is the dominant term i.e.,   
\begin{equation}\label{eq:snbytwodominant}
 \lim_{n \to \infty} \frac{\ln\brc{S\brc{\frac{n}{2}}}}{n} > \lim_{n \to \infty} \frac{\ln\brc{S(\gamma n)}}{n}, \quad 
\gamma \in [0, 1], \gamma \neq \frac{1}{2}, 
\end{equation}
then with positive probability a randomly chosen formula has a satisfying assignment, i.e.
\begin{equation}\label{eq:positiveprobbound}
\lim_{n \to \infty} \prob\brc{N(\alpha, n) > 0} \geq  \frac{2 \sqrt{|B_f(x_k, x_k^2, x_k)|}}{\sigma_s b_q(x_k) \sqrt{k}}, 
\end{equation}
where $x_k$ is the solution of the saddle point equation $a_q(x) = \frac{k}{2}-1$ defined in Lemma 
\ref{lem:moment1hayman}, $a_q(x)$ and $b_q(x)$ are defined according to (\ref{eq:defaqbq}), $B_f(x_k, x_k^2, x_k)$ 
is defined as in Theorem \ref{thm:multisaddle}, and the ``normalized variance'' $\sigma_s^2$ of the 
summation term around $S\brc{\frac{n}{2}}$ is given by 
\begin{equation}\label{eq:defsigmas}
\sigma_s^2 = \frac{1}{4+ \frac{k \dr}{2} ([-1, 1, -1]\cdot B_f(x_k, x_k^2, x_K)^{-1}\cdot [-1, 1, -1]^T) - 8 \dr}. 
\end{equation}
Let $\dr^*$ be the largest literal degree for which $S(n/2)$ is the dominant term, i.e. (\ref{eq:snbytwodominant}) holds, 
then the threshold $\alpha^*$ is lower bounded by 
$
	\alpha^* \geq \alpha_l^* \triangleq \frac {2 \dr^*}{k}.
$
\end{theorem}
\begin{proof} 
From (\ref{eq:moment2deg1}) and Theorem \ref{thm:multisaddle}, the growth rate of $S(\gamma n)$ is given by, 
\begin{multline}\label{eq:defsgamma}
s(\gamma) \triangleq \lim_{n \to \infty} \frac{\ln\brc{S(\gamma n)}}{n} = \\ (1-k \alpha) (\ln(2) + h(\gamma))
+ \alpha \ln\brc{f(t_1, t_2, t_3)} - \dr (1-\gamma) (\ln(t_1)+\ln(t_3)) -\dr \gamma \ln(t_2),  
\end{multline}
where $t_1, t_2, t_3$ is a positive solution of the saddle point equations as defined in 
Theorem~\ref{thm:multisaddle}, 
\begin{multline}\label{eq:saddlepoint}
a_f(\un{t}) \triangleq \left\{ t_1 \frac{\partial \ln\brc{f(t_1, t_2, t_3)}}{\partial t_1}, 
\quad t_2 \frac{\partial \ln\brc{f(t_1, t_2, t_3)}}{\partial t_2}, 
\quad t_3 \frac{\partial \ln\brc{f(t_1, t_2, t_3)}}{\partial t_3} \right\} = \\
\left\{\frac{k}{2} (1-\gamma), \frac{k}{2} \gamma, \frac{k}{2} (1-\gamma) \right\}. 
\end{multline}
In order to compute the maximum exponent of the summation terms, we compute its derivative and equate it to zero, 
\begin{equation}\label{eq:derivatives}
\frac{d s(\gamma)}{d \gamma} = (1-k \alpha) \ln\brc{\frac{1-\gamma}{\gamma}} + \dr \ln(t_1)   
-\dr \ln(t_2) + \dr \ln(t_3) = 0. 
\end{equation}
Note that the derivatives of $t_1, t_2$ and $t_3$ w.r.t. $\gamma$ vanish as they satisfy the saddle point equation. 
Every positive solution $(t_1, t_2, t_3)$ of (\ref{eq:saddlepoint}) satisfies 
$t_1=t_3$  as (\ref{eq:saddlepoint}) and $f(t_1, t_2, t_3)$ are symmetric in 
$t_1$ and $t_3$. If $\gamma = 1/2$ is a maximum, then the vanishing derivative in (\ref{eq:derivatives}) and equality of 
$t_1$ and $t_3$ imply $t_2 = t_1^2$. We substitute $\gamma=1/2$, $t_1=t_3$, and 
$t_2=t_1^2$ in (\ref{eq:saddlepoint}). This reduces (\ref{eq:saddlepoint}) to the saddle point equation corresponding to 
the polynomial $q(x)$ defined in Lemma \ref{lem:moment1hayman} whose solution is 
denote by $x_k$. Then by observing $f(x_k, x_k^2, x_k) = p(x_k)^2$, we have 
\begin{equation}\label{eq:Snbytwo}
S(n/2) = \frac{k^{3/2}}{2^{7/2} \sqrt{\pi n} \sqrt{|B_f(x_k, x_k^2, x_k)|}} 
e^{n \brc{2 \ln(2) (1-k \alpha) + 2 \alpha \ln(p(x_k)) - k \alpha \ln(x_k)}} (1+o(1)). 
\end{equation}
Using the relation that $q(x) = \frac{p(x)}{x}$,  we note that the exponent of $S(n/2)$ is 
twice the exponent of the first moment of the total number of solutions as given in 
(\ref{eq:moment1hayman}). In order to compute the sum over $S(\gamma n)$, we now use 
 Laplace's method, a detailed discussion of which can be found in \cite{Hen74, Deb81, FlS09}.  
We want to approximate the term $S(n/2+\Delta i)$ in terms of $S(n/2)$.  For the coef terms, we make use of the local limit theorem given in 
Theorem \ref{thm:multisaddle} and for the factorial terms we make use of Stirling's approximation. 
This gives,
\begin{equation}\label{eq:sdeltai}
S(n/2 + \Delta i) = S(n/2) e^{-\frac{\Delta i^2}{2 n \sigma_s^2}} (1+o(1)), \quad \text{where }  
\Delta i = O(n^{1/2} ln(n)^{1/3}).
\end{equation} 
Note that in the exponent on the R.H.S. of (\ref{eq:sdeltai}), the linear terms
in $\Delta i$ are absent as the derivative of the exponent vanishes at
$\gamma=1/2$. As the deviation around the term $S(n/2)$ is $\Theta(\sqrt{n})$ and the approximation 
is valid for $\Delta i = O(\sqrt{n} \ln(n)^{1/3})$, the dominant contribution comes from 
$-\Theta(\sqrt{n}) \leq \Delta i \leq \Theta(\sqrt(n))$. 
We are now ready to obtain the estimate for the second moment. 
\begin{eqnarray}
E(N^2(\alpha, n))  & \stackrel{(\ref{eq:sdeltai})}{=} & \sum_{\Delta i= - c \sqrt{n}}^{c \sqrt{n}} S(n/2) e^{-\frac{\Delta i^2}{2 n \sigma_s^2}} (1+o(1)), \\
\label{eq:sumtoint} & = & S(n/2) \int_{\delta=-\infty}^{\infty} e^{-\frac{\delta^2}{2 n \sigma_s^2}} d \delta (1+o(1)). \\ 
\label{eq:moment2approx} & = & S(n/2) \sqrt{2 \pi n \sigma_s^2} (1+o(1)).
\end{eqnarray}
We can replace the sum by an integral by choosing sufficiently large $c$. Using the second moment method given in (\ref{eq:smmethod}) and combining Lemma \ref{lem:moment1hayman}, (\ref{eq:Snbytwo}), and 
(\ref{eq:moment2approx}), we obtain 
\begin{equation}
\prob\brc{N(\alpha, n) > 0}  \geq  \frac{E(N(\alpha, n)^2)}{E(N(\alpha, n)^2)} 
  = \frac{2 \sqrt{|B_f(x_k, x_k^2, x_k)|}}{\sigma_s b_q(x_k) \sqrt{k}} (1+o(1)).
\end{equation}
Letting $n$ go to infinity, we obtain (\ref{eq:positiveprobbound}).  
Clearly, if the supremum of the growth rate of $S(\gamma n)$ is not achieved at $\gamma = 1/2$, then the lower bound given 
by the second moment method converges to zero. This gives the desired lower bound on the threshold. 
\end{proof}

We can easily extend this result to the 2-regular case. In the following theorem we accomplish 
this task. Due to space limitation, we omit explanation of some steps which can be found in \cite{RARS10t}.  
\begin{theorem}\label{thm:2reglowerbound}
Consider the 2-regular random k-SAT model where the number of variables with degree $\dr$ (resp. $\dr+1$) is 
$n_\dr = \Lambda_\dr n$ (resp. $n_{\dr+1} = \Lambda_{\dr+1} n$). Let $S(i_\dr, i_{\dr+1}) \triangleq
S(\gamma_\dr n_\dr, \gamma_{\dr+1} n_{\dr+1})$ be the summation term on the
R.H.S. of (\ref{eq:moment2}) corresponding to overlap $i_\dr$(resp.
$i_{\dr+1}$) on the degree $\dr$(resp. $\dr+1$) literals. Let $g(\gamma_\dr, \gamma_{\dr+1})$ 
be the growth rate of $S(\gamma_\dr n_\dr, \gamma_{\dr+1} n_{\dr+1})$ i.e. 
$g(\gamma_\dr, \gamma_{\dr+1}) \triangleq \lim_{n \to \infty} \frac{\ln\brc{S(n_\dr \gamma_\dr, n_{\dr+1})}}{n}$. 
If 
\[
g\brc{\frac{1}{2}, \frac{1}{2}} >  g(\gamma_\dr, \gamma_{\dr+1}), \gamma_\dr \in [0, 1], \gamma_{\dr+1} \in [0, 1], 
\gamma_\dr \neq \frac{1}{2}, \gamma_{\dr+1} \neq \frac{1}{2}, 
\] 
then with positive probability a randomly chosen formula has a solution. More precisely,
\begin{equation}\label{eq:positiveprobbound2reg}
\lim_{n\to \infty} \prob\brc{N(\alpha, n) > 0} \geq \frac{\sqrt{|B_f(x_k, x_k^2, x_k)| \Lambda_\dr \Lambda_{\dr+1}}}{b_q(x_k) \sqrt{k |\Sigma|}}.
\end{equation}
The definition of $x_k$, $b_q(x_k)$, and $B_f(x_k, x_k^2, x_k)$ is same as in the Theorem \ref{thm:reglowerbound}. The $2 \times 2$ matrix $\Sigma$ is defined 
via, 
 \begin{multline}\label{eq:defSigma}
C_f = [-1, 1, -1].(B_f(x_k, x_k^2, x_k))^{-1}.[-1, 1, -1]^T, \quad 
A = \frac{4}{\Lambda_\dr} + 2 \dr^2 \brc{\frac{C_f}{2 \alpha}-\frac{4}{k \alpha}}, \\
B = \frac{4}{\Lambda_{\dr+1}} + \frac{2 (\dr+1)^2}{\alpha} \brc{\frac{C_f}{2}-\frac{4}{k}}, \hspace*{0.02in}
C = \frac{2 \dr (\dr+1)}{\alpha} \brc{\frac{C_f}{2}-\frac{4}{k}}, \hspace*{0.02in} \text{then} \hspace*{0.02in} \Sigma = \begin{bmatrix}A & B \\ B & C \end{bmatrix}^{-1}. 
\end{multline}
The threshold $\alpha^*$ is lower bounded by $\alpha_l^*$, where $\alpha_l^*$ is defined by
\[
\alpha_l^* = \sup\left\{\alpha: g\brc{\frac{1}{2}, \frac{1}{2}} >  g(\gamma_\dr, \gamma_{\dr+1}), \gamma_\dr \in [0, 1], \gamma_{\dr+1} \in [0, 1], 
\gamma_\dr \neq \frac{1}{2}, \gamma_{\dr+1} \neq \frac{1}{2} \right\}.
\]
\end{theorem}
\begin{proof}
Define $\Gamma(\gamma_\dr, \gamma_{\dr+1}) = 
\dr \Lambda_\dr \gamma_\dr + (\dr+1) \Lambda_{\dr+1} \gamma_{\dr+1}$. Then by using Theorem \ref{thm:multisaddle} and Stirling's approximation, we obtain  
\begin{multline}
g(\gamma_\dr, \gamma_{\dr+1})  = \ln(2) +  
\Lambda_\dr h(\gamma_\dr) + \Lambda_{\dr+1} h(\gamma_{\dr+1})  \\ 
+ \brc{k \alpha - 2 \Gamma(\gamma_\dr, \gamma_{\dr+1})} \ln\brc{\frac{k \alpha}{2} - \Gamma(\gamma_\dr, \gamma_{\dr+1})} \\ 
+ 2 \brc{\Gamma(\gamma_\dr, \gamma_{\dr+1})} \ln\brc{\Gamma(\gamma_\dr, \gamma_{\dr+1})}  
- k \alpha \ln\brc{k \alpha} \\ + \alpha \ln\brc{f(\un(t))} 
- \brc{\frac{k \alpha}{2} - \Gamma(\gamma_\dr, \gamma_{\dr+1})} \ln\brc{t_1 t_3} - \brc{\Gamma(\gamma_\dr, \gamma_{\dr+1})} \ln(t_2),
\end{multline}
where $\un{t}=\{t_1, t_2, t_3\}$ is a positive solution of the saddle point point equations as given in Theorem \ref{thm:multisaddle},  
\begin{equation}\label{eq:saddlepoint2deg}
a_f(\un{t}) = \left\{ \frac{k}{2} - \frac{\Gamma(\gamma_\dr, \gamma_{\dr+1})}{\alpha}, \frac{\Gamma(\gamma_\dr, \gamma_{\dr+1})}{\alpha}, 
 \frac{k}{2} - \frac{\Gamma(\gamma_\dr, \gamma_{\dr+1})}{\alpha} \right\}, 
\end{equation}
corresponding to the coefficient term of power of $f(x_1, x_2, x_3)$. 
In order to obtain the maximum exponent, we take the partial derivatives of $g(\gamma_\dr, \gamma_{\dr+1})$ with 
respect to $\gamma_{\dr}$ and $\gamma_{\dr+1}$ and equate them to zero. This gives the following equations.
\begin{multline}\label{eq:derivative2deg}
\ln\brc{\frac{1-\gamma_\dr}{\gamma_\dr}} - 2 \dr \ln\brc{\frac{k \alpha}{2}-\Gamma(\gamma_\dr, \gamma_{\dr+1})} + 
2 \dr \ln\brc{\Gamma(\gamma_\dr, \gamma_{\dr+1})} + \dr \ln\brc{\frac{t_1 t_3}{t_2}} = 0, \\ 
\hspace*{-0.06in} \ln\brc{\frac{1-\gamma_{\dr+1}}{\gamma_{\dr+1}}} - 2 (\dr+1) \ln\brc{\frac{k \alpha}{2}-\Gamma(\gamma_\dr, \gamma_{\dr+1})} \\
+ 2 (\dr+1) \ln\brc{\Gamma(\gamma_\dr, \gamma_{\dr+1}} + (\dr+1) \ln\brc{\frac{t_1 t_3}{t_2}}  = 0.
\end{multline}
Note that $t_1=t_3=x_k$, $t_2=x_k^2$, $\gamma_\dr=\frac{1}{2}$, and $\gamma_{\dr+1} = \frac{1}{2}$ is a solution of 
(\ref{eq:saddlepoint2deg}), (\ref{eq:derivative2deg}), which corresponds to $S\brc{\frac{n_\dr}{2}, \frac{n_{\dr+1}}{2}}$, 
where 
$x_k$ is the solution of the saddle point equation corresponding to $q(x)$ defined in Lemma \ref{lem:moment1hayman}. 
We recall that for the second moment method to work, the maximum exponent should be equal to twice the 
exponent of the average number of solutions. Indeed by the proposed solution, the term $S\brc{\frac{n_\dr}{2}, \frac{n_{\dr+1}}{2}}$ has an exponent 
which is twice that of the average number of solutions. If this is also the maximum, then we have the desired 
result. Assuming that $S\brc{\frac{n_\dr}{2}, \frac{n_{\dr+1}}{2}}$ has the maximum exponent, we now compute the 
second moment of the total number of solutions. By using Stirling's approximation and the local limit result of 
Theorem \ref{thm:multisaddle}, we obtain
\begin{equation}\label{eq:ratio2deg}
\frac{S\brc{i_{\dr}+\Delta i_\dr, i_{\dr+1} + \Delta i_{\dr+1}}}{S\brc{i_\dr, i_{\dr+1}}} = 
e^{- \frac{1}{2 n} [\Delta i_\dr, \Delta i_{\dr+1}]\cdot\Sigma^{-1}\cdot[\Delta i_\dr, \Delta i_{\dr+1}]^T} (1+o(1)), 
\end{equation}
where the matrix $\Sigma$ is defined in (\ref{eq:defSigma}). Using the same series of arguments as in Theorem \ref{thm:reglowerbound}, we obtain  
\begin{align}
E\brc{N(\alpha, n)^2} = \sum_{\Delta i_\dr, \Delta i_{\dr+1}} S\brc{\frac{n_\dr}{2}, \frac{n_{\dr+1}}{2}} e^{- \frac{1}{2 n} [\Delta i_\dr, \Delta i_{\dr+1}]\cdot\Sigma^{-1}\cdot[\Delta i_\dr, \Delta i_{\dr+1}]^T} (1+o(1)), \\
 = S\brc{\frac{n_\dr}{2}, \frac{n_{\dr+1}}{2}}  \int_{-\infty}^\infty \int_{-\infty}^\infty e^{- \frac{1}{2 n} [x_\dr, x_{\dr+1}]\cdot\Sigma^{-1}\cdot[x_\dr, x_{\dr+1}]^T} dx_\dr dx_{\dr+1}(1+o(1)), \\
= \frac{\sqrt{|\Sigma|} k^{\frac{3}{2}}}{4 \sqrt{|B_f(x_k, x_k^2, x_k)| \Lambda_\dr \Lambda_{\dr+1}}} 
e^{2 n \brc{\ln(2) - k \alpha \ln(2) + \alpha \ln\brc{p(x_k)}-\frac{k \alpha}{2} \ln\brc{x_k}}} (1+o(1)).
\end{align}
By using the second moment method, we obtain the bound given in (\ref{eq:positiveprobbound}). 
Note that the second moment method fails if the term $S\brc{\frac{n_\dr}{2}, \frac{n_{\dr+1}}{2}}$ is not the 
dominant term. This gives the lower bound $\alpha^*_l$ on $\alpha^*$. 
\end{proof}

In the next section we discuss the obtained lower and upper bounds on the satisfiability 
threshold and NAE-satisfiability threshold. 
\section{Bounds on Threshold}
In Table \ref{tb:regthreshold}, lower bounds and upper bounds for the 
satisfiability threshold are given.  The upper bound is computed by the first
moment method. As expected, we obtain the same upper bound for regular random
$3$-SAT as given in \cite{BDIS05}. The lower bound is derived by the second
moment method for strictly regular random $k$-SAT. In order to apply the
second moment method, we have to verify that $s(\gamma)$, defined in
(\ref{eq:defsgamma}), attains its maximum at $\gamma=\frac{1}{2}$ over the unit
interval.  This requires that $\gamma=\frac{1}{2}$ is a positive 
solution of the system of equations consisting of (\ref{eq:saddlepoint}) and 
(\ref{eq:derivatives}) and it corresponds to a global maximum over $\gamma \in [0, 1]$. Also, $\sigma_s^2$ defined in
(\ref{eq:defsigmas}) should be positive. The system of equations
(\ref{eq:saddlepoint}), (\ref{eq:derivatives}) is equivalent to a system of
polynomial equations. For small value of $k$, we can solve this system of
polynomial equations and verify the desired conditions. In Table 
\ref{tb:regthreshold} this has been done for $k=3, 4$. The obtained lower bound 
for $3$-SAT is $2.667$ which is an improvement over the algorithmic lower bound 
$2.46$ given in \cite{BDIS05}. For larger values of
$k$, the degree of monomials in (\ref{eq:derivatives}) grows exponentially in
$k$. Thus, solving (\ref{eq:saddlepoint}) and  (\ref{eq:derivatives}) becomes
computationally difficult. However, $s(\gamma)$ can be easily computed as its
computation requires solving only (\ref{eq:saddlepoint}), where the maximum
monomial degree is only $k$. Thus, the desired condition for maximum of
$s(\gamma)$ at $\gamma=\frac{1}{2}$ can be verified numerically in an efficient
manner.  

Note that the difference between the lower bound obtained by applying the 
second moment method to the strictly regular case can differ by at most $2/k$ 
from the corresponding lower bound for the $2$-regular case. We observe that 
as $k$ increases the lower bound seems to converge to $2^k \ln(2) - (k+1) \frac{\ln(2)}{2}-1$, 
which is the lower bound for the uniform model. 

We observe similar behavior for the NAE-satisfiability bounds. As expected, we observe that the upper bound on the NAE-satisfiability threshold for the regular random model 
converges to $2^{k-1} \ln(2)-\frac{\ln(2)}{2}$. The lower bound obtained by applying the 
second moment method to the regular random model seems to converge to the value obtained for 
the uniform model. Thus, the NAE-threshold for the regular random model is  $2^{k-1} \ln(2)-O(1)$. 
This suggests that as $k$ increases, the threshold of the regular model does not differ much from the uniform model. 
\begin{table}[h]
\begin{center}
\begin{tabular}{|c|c|c|c|c|c|c|c|}
\hline 
$k$ & $\dr^*$ & $\alpha_l^*$ & $\alpha_u^*$ & $\alpha_{l, \mathrm{uni}}^*-\alpha_l^*$ & $\dr^*_{\mathrm{\tiny NAE}}$ & $\alpha_{l, \mathrm{\tiny NAE}}^*$ & $\alpha_{u, \mathrm{\tiny NAE}}^*$\\
\hline
$3$ & $4$ & $2.667$ & $3.78222$ & $0.492216$ & $3$ & $2$ & $2.40942$\\
\hline
$4$ & $16$ & $8$ & $9.10776$ & $0.357487$  & $8$ & $4$ & $5.19089$\\
\hline
$7$ & $296$ & $84.571$ & $85.8791$ & $0.378822$ & $152$ & $43.4286$ & $44.0139$\\
\hline
$10$ & $3524$ & $704.8$ & $705.9533$ & $0.170403$ & $1770$ & $354$ & $354.545$\\
\hline
$15$ & $170298$ & $22706.4$ & $22707.5$ & $0.101635$ & $85167$ & $11355.6$ & $11356.2$\\
\hline
$17$ & $772182$ & $90844.94$ & $90845.9$ & $0.007749$ & $386114$ & $45425.2$ & $45425.7$\\
\hline
\end{tabular} 
\end{center}
\caption{Bounds on the satisfiability threshold for strictly regular random
$k$-SAT. $\alpha^*_l$ and $\dr^*$ are defined in Theorem
\ref{thm:reglowerbound}. The upper bound $\alpha_u^*$ is obtained by the first moment method. 
$\alpha_{l, \mathrm{uni}}^* = 2^k \ln(2) - (k+1) \frac{\ln(2)}{2}-1$ is lower bound for uniform model obtained in \cite{AcP04}.
The quantities $\dr_{\mathrm{\tiny NAE}}^*, \alpha_{l, \mathrm{\tiny NAE}}, \alpha_{u, \mathrm{\tiny NAE}}$ are analogously defined 
for the NAE-satisfiability. 
}\label{tb:regthreshold} 
\end{table}

Our immediate future work is to derive explicit lower bounds for the 
regular random $k$-SAT model for large values of $k$ as was done for the 
uniform model in \cite{AcP04, AcM06}. The challenge is that the 
function $s(\gamma)$ depends on the solution of the system of polynomial 
equations given in (\ref{eq:saddlepoint}). Thus determining the maximum requires 
determining the behavior of the positive solution of this system of polynomial 
equations. Another interesting direction is the maximum 
satisfiability of regular random formulas. For the uniform model, the maximum satisfiability 
problem was addressed in \cite{ANP07} using the second moment method. In \cite{RARS10maxsat}, 
authors have derived lower and upper bounds on the maximum satisfiability threshold 
of regular random formulas.  

\bibliographystyle{siam}
\bibliography{ksat} 
\end{document}